\documentclass[letterpaper]{article}

\usepackage{authblk}
\usepackage{graphicx}
\usepackage{subfig}
\usepackage{amsfonts}
\usepackage{multirow}
\usepackage{amssymb}
\usepackage{amsmath}
\usepackage{amsthm}

\begin{document}

\newtheorem{theorem}{Theorem}
\newtheorem{example}{Example}
\newtheorem{lemma}{Lemma}
\newtheorem{corollary}{Corollary}
\newtheorem{definition}{Definition}

\title{Oblivious Stacking and MAX $k$-CUT for Circle Graphs}

\author{Martin Olsen}
\affil{Department of Business Development and Technology \authorcr Aarhus University \authorcr Denmark \authorcr martino@btech.au.dk}

\maketitle

\begin{abstract}
Stacking is an important process within logistics. Some notable examples of items to be stacked are steel bars or steel plates in a steel yard or containers in a container terminal or on a ship. We say that two items are {\em conflicting} if their storage time intervals overlap in which case one of the items needs to be rehandled if the items are stored at the same LIFO storage location. We consider the problem of stacking items using $k$ LIFO locations with a minimum number of conflicts between items sharing a location. We present an extremely simple online stacking algorithm that is oblivious to the storage time intervals and storage locations of {\em all other items} when it picks a storage location for an item. The risk of assigning the same storage location to two conflicting items is proved to be of the order $1/k^2$ under mild assumptions on the distribution of the storage time intervals for the items. Intuitively, it seems natural to pick a storage location uniformly at random in the oblivious setting implying a risk of $1/k$ so the risk for our algorithm is surprisingly low. Our results can also be expressed within the context of the MAX $k$-CUT problem for circle graphs. The results indicate that circle graphs on average have relatively big $k$-cuts compared to the total number of edges.  
%\keywords{Oblivious algorithms \and MAX $k$-CUT \and Stacking \and Circle graphs}
\end{abstract}

\section{Introduction}

We consider a storage area with a continuous flow of items entering and leaving the area. The storage area consists of $k$ storage locations accessible in a LIFO manner. If we want to retrieve a target item from a storage location, we have to deal with all the items present at the storage location that have arrived later than the targeted item. The items that have arrived later but have not left yet are conflicting with the targeted item so we focus on the problem of assigning storage locations to the items with a minimum number of conflicts among items sharing a location. We refer to this problem as {\em stacking} since each storage location acts as a stack data structure.

The stacking problem is an important problem within logistics with many applications such as shipment of containers~\cite{Borgman2010,kshift-dam}, track assignment for trains (the storage locations are train tracks)~\cite{Cornelsen2007,Demange2012} and stacking of steel bars or steel plates~\cite{Kim11,Konig2007,Rei13}. Actually, a storage location can be any place where items potentially block each other as described earlier: a truck, a train wagon, a shelve in a warehouse, etc.

We can rephrase the problem as a graph problem if we construct a graph with a vertex for each item and an edge between two vertices if the corresponding items are conflicting. Now the objective is to color the vertices using $k$ colors with a minimum number of edges connecting two vertices with the same color or, equivalently, to maximize the cardinality of the set of edges connecting vertices with different colors. The latter set of edges is typically called the {\em cut} and the problem is commonly known as the MAX $k$-CUT problem. The MAX $k$-CUT problem is a famous graph problem appearing in many contexts.
The MAX $k$-CUT problem is closely related to the coloring problem where the objective is to use as few colors as possible with no edges connecting vertices with the same color. For the MAX $k$-CUT problem the number of colors is known a priori as opposed to the coloring problem where the number of colors is output by the algorithm.

\subsection{Related Work}
%k-Coloring of circle graphs

An {\em offline} stacking algorithm has access to all information on all the items before the assignment of storage locations is carried out in contrast to an {\em online} stacking algorithm where a decision on where to store an item is made without access to any information on future items.

First, we consider related work for the coloring version of stacking where the objective is to use as few storage locations as possible with no conflicts. Upper and lower bounds for the competitive ratio for online coloring are presented by Demange et al.~\cite{Demange2012} and Demange and Olsen~\cite{DBLP:conf/walcom/DemangeO18}. Olsen and Gross~\cite{DBLP:conf/iccl2/OlsenG15} have developed a polynomial time algorithm for online coloring with a competitive ratio that converges to $1$ in probability if the endpoints of the storage time intervals are picked independently and uniformly at random. Olsen~\cite{OlsenLion20} has also shown how to use Reinforcement Learning to improve online stacking heuristics.

The offline version of coloring is NP-hard for unbounded stack capacity~\cite{Avriel2000} and for fixed stack capacity $h \geq 6$~\cite{Cornelsen2007} and the computational complexity for $2 \leq h \leq 5$ is an open problem (to the best of our knowledge).

%Max k-Cut for circle graphs 
We now take a look at related work for the MAX $k$-CUT version of stacking where the number of storage locations is given and where the objective is to minimize the number of conflicts among the items. Handling shipping containers that block each other in a stack is known as {\em rehandling} or {\em shifting}. Tierney et al.~\cite{kshift-dam} show how to compute the minimum number of shifts offline in polynomial time using a fixed number of stacks (storage locations) with bounded capacity. The MAX $k$-CUT version of stacking is at least as hard as the coloring problem as can be easily seen by reduction from the coloring problem. The offline version is even NP-hard for $k=2$ (unbounded capacity)~\cite{Pocai16}.

For the offline version of MAX $k$-CUT for graphs in general, we can achieve an approximation ratio $1-\frac{1}{k}$
with a simple polynomial time algorithm but it is not possible to achieve an approximation ratio better than $1-\frac{1}{34k}$ if $NP \neq P$~\cite{cj97-02}. Finally, it should be mentioned that Coja-Oghlan et al.~\cite{Oghlan06} examine the MAX $k$-CUT problem in random graphs.

\subsection{Contribution}\label{sec:contribution}
%Dobbelt: stacking og random circle graphs ...

For the remaining part of the paper, we focus on the MAX $k$-CUT version of stacking so from now on we assume that our goal is to minimize the number of conflicts between items stored at the same location using $k$ storage locations.

Our main aim is to investigate the case where an online algorithm does {\em not} have access to the storage time intervals for the items already stored and does {\em not} keep track of the locations of these items but {\em only} has access to the arrival time and the departure time for an entering item when it has to pick a storage location for that item. In other words, the algorithm is {\em oblivious} to the time intervals for {\em all other items} and the storage locations assigned to items in the past. We will refer to such an algorithm as an {\em oblivious stacking algorithm}. Such an algorithm is extremely simple to implement and it can even be used in a scenario with several disconnected stacking agents. To the best of our knowledge, we are the first to examine this type of stacking algorithms but such oblivious algorithms have been examined for other problems (for example routing~\cite{Bansal03,DBLP:journals/tc/BuschMX08}).

Intuitively, it might look a little strange to consider the oblivious setting since it seems hard to do anything better than to assign a random storage location to an entering item if we do not have any information on the other items. The risk of assigning the same storage location to two conflicting items is $1/k$ if we use the random strategy. We present an extremely simple algorithm for oblivious stacking where this risk is reduced dramatically to the order $1/k^2$ under the assumption that the storage time intervals are produced by a model proposed by Scheinerman~\cite{SCHEINERMAN1990287}. The algorithm is based on the intuition that the risk of a conflict between two items is low if the {\em centers} of their storage time intervals are close to each other since such items only will have a conflict if the lengths of their storage time intervals are roughly the same. The basic principle of our algorithm is to place items at the same storage location if the centers of their storage time intervals are close to each other or far apart.

Our algorithm can used on its own but it can also be used as a part of a hierarchical approach to divide the items into $k$ groups with only a few conflicts in each group. We can for example divide a steel yard or a container terminal into a few areas and then use our algorithm to assign an area to each item in a preprocessing step. For each area, we then apply a local stacking algorithm to assign a specific storage location to each item (the local algorithms do not have to be identical).

The stacking conflict graphs described earlier are so called {\em circle graphs} -- also known as {\em interval overlap graphs} (more details in Sec.~\ref{sec:preliminaries}). From a graph theory perspective, we offer a contribution for an audience interested in circle graphs and show that our algorithm computes a $k$-cut with expected cardinality satisfying $E(|cut|)\geq(1-O(1/k^2))E(m)$ for random circle graphs represented by an extended version of the Scheinerman model ($m$ denotes the number of edges). This indicates that random circle graphs on average have relatively big $k$-cuts compared to the number of edges $m$ in contrast to random graphs of the Erd\H{o}s-R\'enyi type~\cite{Oghlan06}.

The problem that we consider is formally defined in Sec.~\ref{sec:preliminaries} where we also present our algorithm and the first model for generating the instances. The performance of our algorithm is analyzed in Sec.~\ref{sec:Scheinerman} using average case analysis where we compute the exact risk that our oblivious algorithm assigns two conflicting items to the same storage location. In Sec.~\ref{sec:extended_Scheinerman} we extend our model for generating the instances and demonstrate that the risk is low (of the order $1/k^2$) even in a more generic setting. Finally, our results are related to circle graphs/interval overlap graphs.

\section{Preliminaries}\label{sec:preliminaries}

\subsection{The Problem}

Two intervals $[a, b]$ and $[c, d]$ are said to {\em overlap} if and only if the intervals intersect and neither is contained in the other: $a < c < b < d$ or $c < a < d < b$. Two items are conflicting if and only if their storage time intervals overlap since shifting/rehandling is necessary exactly in this case if the items are stored at the same location. A graph with vertices representing intervals and edges representing overlaps is commonly known as an {\em interval overlap graph}. A {\em circle graph} is a graph where vertices represent chords of a circle with an edge between two vertices if the corresponding chords intersect. Gavril~\cite{Gavril1973} has shown that a graph is an interval overlap graph if and only if it is a circle graph. In other words, a conflict graph for stacking is an interval overlap graph/circle graph implying relevance of our results, as already mentioned, for a broader audience interested in the MAX $k$-CUT problem for such graphs. It should be noted that the connection between stacking and interval overlap graphs/circle graphs was established by Avriel et al.~\cite{Avriel2000}.

To sum up, the problem considered in this paper can be formally and concisely defined as follows where the items are represented by their storage time intervals:
\begin{definition}\label{def:max_k_cut} The MAX $k$-CUT STACKING
  problem:
\begin{itemize}
\item Instance: A set of $n$ intervals $\mathcal{I}_n = \{I_1, I_2, \ldots , I_n\}$
\item Solution: A coloring of the intervals using colors $\{1, 2, 3, \ldots , k\}$ with a maximum cardinality cut where the cut is the set of unordered pairs of overlapping intervals with different colors
\end{itemize}
\end{definition}

An example of a MAX $k$-CUT STACKING instance and an optimal solution is displayed in Fig.~\ref{fig:example} for $n=4$ and $k=2$. The intervals are shown to the left and the corresponding interval overlap graph/circle graph is shown to the right. The intevals/vertices have received a gray or a black color producing a cut with cardinality 4: $|cut|=4$. The optimal solution has only one conflict between items stored at the same location.

\begin{figure}
  \centering
  \subfloat{\label{fig:intervals_example}\includegraphics[scale=0.55]{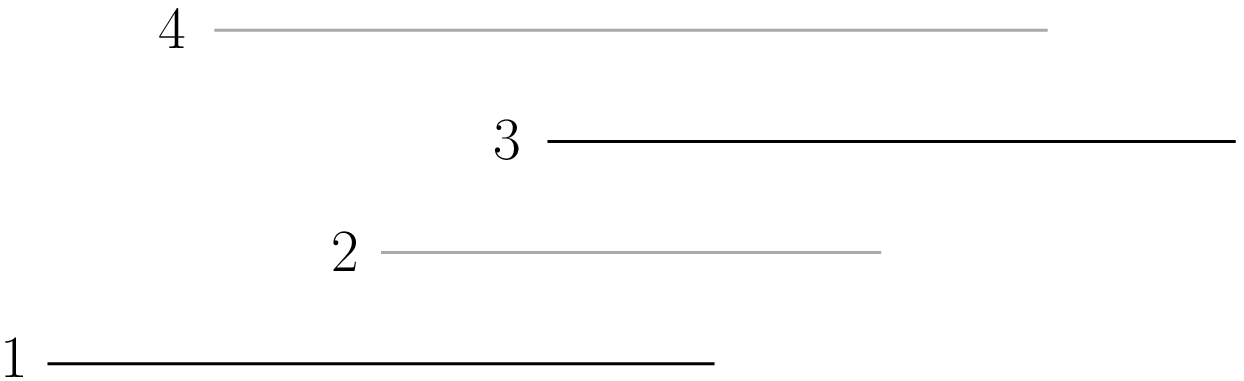}}
  \hspace{.5cm}
  \subfloat{\label{fig:graph_example}\includegraphics[scale=0.55]{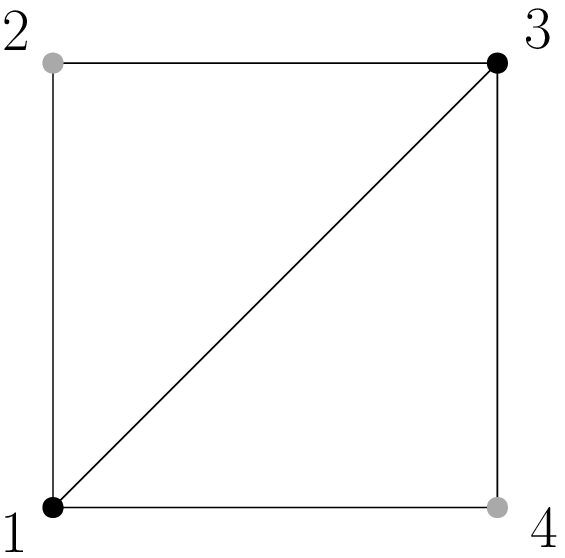}}
  \caption{An example of a MAX $k$-CUT STACKING instance for $n=4$ and $k=2$ and an optimal solution. The coloring uses the colors gray and black and the conflict graph is shown to the right. The cut $\{\{1, 2\}, \{2, 3\}, \{1, 4\}, \{3, 4\}\}$ has cardinality 4: $|cut|=4$.}
\label{fig:example}
\end{figure}
Please note that there is a one-to-one correspondence between the colors and the storage locations. It should also be noted that we consider storage locations with unbounded capacity. According to Kim et al.~\cite{Kim11}, this assumption is not critical considering stacking of steel plates since many plates can be stacked together because the thicknesses of the plates are very small. If our algorithm is used in a hierarchical approach as described in Sec.~\ref{sec:contribution} then it also makes sense to assume that the capacity is unbounded so there are applications where this assumption is justified.

\subsection{The Instance Model}

We will use a simple stochastic model for generating the instances introduced by Scheinerman~\cite{SCHEINERMAN1990287} in his work on random interval (intersection\footnote{Two overlapping intervals intersect but the converse is not necessarily true}) graphs:

\begin{definition}\label{def:Scheinerman_model} The Scheinerman Model~\cite{SCHEINERMAN1990287}: The centers of the intervals in $\mathcal{I}_n$ are drawn independently using a uniform distribution on $[0, 1]$ and the lengths of the intervals are drawn independently using a uniform distribution on $[0, L]$ for some number $L$.
\end{definition}

In order to be able obtain exact results for the analysis to follow, we assume the following:
\begin{equation}
\label{eq:L_assumption}
1/\left(\frac{k}{k-1} L\right) \in \mathbb{Z} \enspace .
\end{equation}
Our second model for generating instances (defined later in Sec.~\ref{sec:extended_Scheinerman}) allows us to use an arbitrary bounded continuous probability density function for the lengths where $L$ denotes an upper bound of the lengths so the second model is more flexible than our first model but qualitatively the results obtained for the two models are the same. As already stated, the assumption (\ref{eq:L_assumption}) facilitates exact computations.

\subsection{The Algorithm}

Our algorithm works as follows. The interval $[0, 1]$ is split into $\frac{k-1}{L}$ consecutive intervals $J_1, J_2, \ldots, J_{\frac{k-1}{L}}$ with length $\frac{L}{k-1}$. We now color the intervals with the colors $\{1, 2, \ldots, k\}$ in a circular manner such that $J_i$ receives color $(i-1) \bmod k + 1$. Please note that (\ref{eq:L_assumption}) implies that the final interval receives the color $k$ and has $1$ as its right endpoint.

Let an item represented by the storage time interval $I = [x, y]$ with center $c(I) = (x+y)/2$ enter the storage area. The algorithm locates the $J$-interval containing the center and assign the color of this $J$-interval to the entering item:
\begin{equation}
\label{eq:algorithm}
A(I) = \left\lfloor \frac{(k-1) c(I)}{L} \right\rfloor \bmod k + 1, \enspace I \in \mathcal{I}_n \enspace .
\end{equation}
The dynamics of the algorithm is illustrated in Fig.~\ref{fig:algorithm} showing an example with $k=3$ and $L=1/3$.
\begin{figure}
  \centering 
  \includegraphics[scale=1.0]{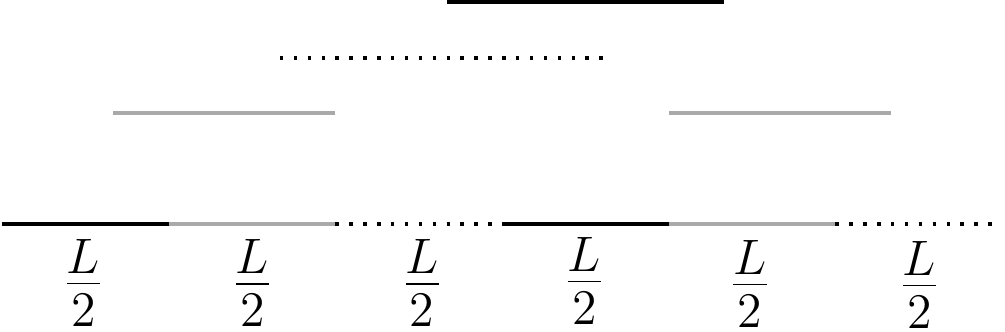}
  \caption{The figure shows how our algorithm works for $k=3$ and $L=1/3$. The items are represented by the four storage time intervals at the top. An item receives the color of the $J$-interval at the bottom containing the center of the storage time interval for the item.}
  \label{fig:algorithm}
\end{figure}
This is a very simple online algorithm that only considers information on the current interval. As earlier mentioned, we refer to such an algorithm as an {\em oblivious} algorithm. The running time per interval is $O(1)$.

The intuition behind our algorithm is as follows. If two intervals receive the same color then there are two possibilities: 1) The centers of the intervals belong to different $J$-intervals in which case the intervals are not overlapping, or 2) The centers of the intervals belong to the same $J$-interval in which case the centers are close to each other implying a low risk of an overlap since an overlap only occurs if the lengths of the intervals are roughly the same. In other words, we will observe relatively few overlapping intervals with the same color.

\section{Analysis of the Scheinerman Model}\label{sec:Scheinerman}

The intuition will now be verified using exact computations.
The key observation is as follows: two intervals are not likely to overlap if their centers are close to each other. We now present a lemma quantifying this observation based on the Scheinerman model (Definition~\ref{def:Scheinerman_model}). Let $C=d$ denote the event that the distance between the centers of two intervals is $d$ and let $OV$ denote the event that two intervals overlap. 

\begin{lemma}\label{lemma:P_ov_given_center} For two intervals drawn using the Scheinerman model we have the following:
$$\Pr(OV \mid C=xL) = \begin{cases} 4x - 6x^2 & ,\enspace 0 \leq x \leq 0.5 \\
                      \frac{1}{2}(2-2x)^2    & ,\enspace 0.5 < x \leq 1 \\     
                      0                      & ,\enspace 1 < x %
        \end{cases}$$
\end{lemma}

\begin{proof}

\begin{figure}
  \centering 
  \includegraphics[scale=0.9]{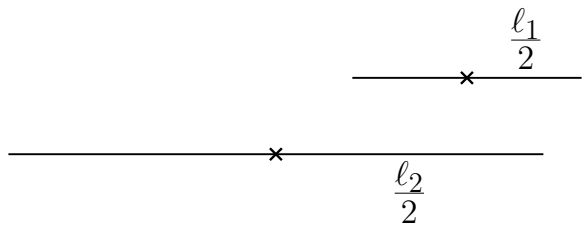}
  \caption{Two intervals with lengths $\ell 1 < \ell_2$ and distance $d$ between the centers overlap if and only if $d < \frac{\ell_1}{2} + \frac{\ell_2}{2} \wedge d + \frac{\ell_1}{2} > \frac{\ell_2}{2}$.}
  \label{fig:P_ov_given_center_two_intervals}
\end{figure}

Consider two intervals with lengths $\ell_1$ and $\ell_2$ with $\ell_1 < \ell_2$ (see Fig.~\ref{fig:P_ov_given_center_two_intervals}). Let $d$ denote the distance between the centers of the intervals. The intervals overlap if and only if
\begin{equation}
\label{eq:overlap_condition}
d < \frac{\ell_1}{2} + \frac{\ell_2}{2} \wedge d + \frac{\ell_1}{2} > \frac{\ell_2}{2} \enspace .
\end{equation}

\begin{figure}
  \centering 
  \includegraphics[scale=0.8]{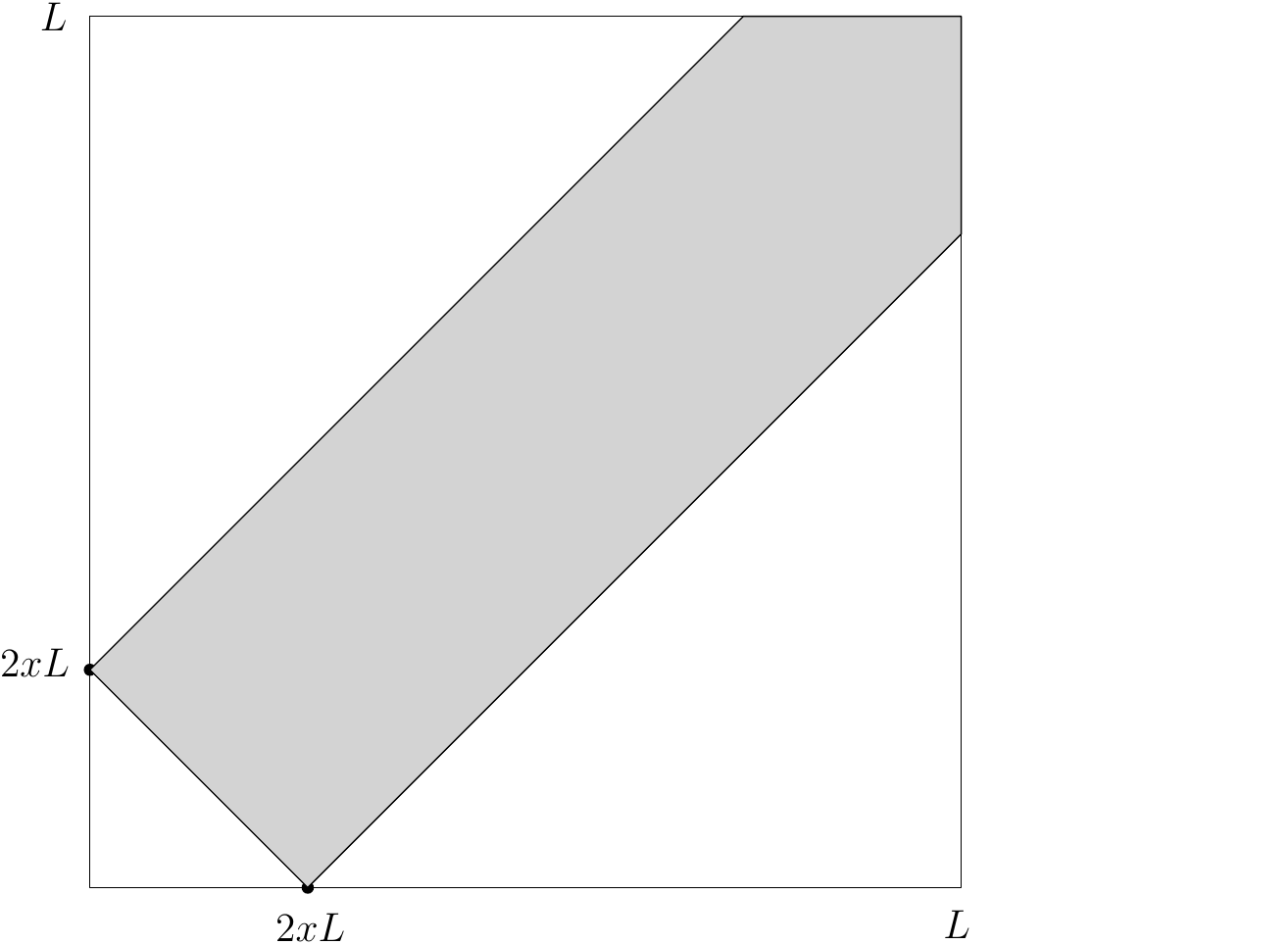}
  \caption{The square $[0, L] \times [0, L]$ represents the possible pairs of lengths for the two intervals. The gray region illustrates the cases where the two intervals overlap given $C = xL$ for $x \leq 0.5$.}
  \label{fig:P_ov_given_center}
\end{figure}

The gray region in Fig.~\ref{fig:P_ov_given_center} contains all pairs $(\ell_1, \ell_2)$ -- including intervals with $\ell_1 > \ell_2$ -- that correspond to overlapping intervals with $d = xL$ for $x \leq 0.5$. The conditional probability $\Pr(OV \mid C=xL)$ is computed as the area of the gray region divided by $L^2$: 
$$\Pr(OV \mid C=xL) = 1-2x^2-(1-2x)^2 = 4x - 6x^2 \enspace, \enspace x \leq 0.5 \enspace .$$
The case $0.5 < x \leq 1$ is handled in a similar way (see Fig.~\ref{fig:P_ov_given_center_case_B}):
$$\Pr(OV \mid C=x L) = \frac{1}{2}(2-2x)^2, \enspace 0.5 < x \leq 1 \enspace .$$
\begin{figure}
  \centering 
  \includegraphics[scale=0.8]{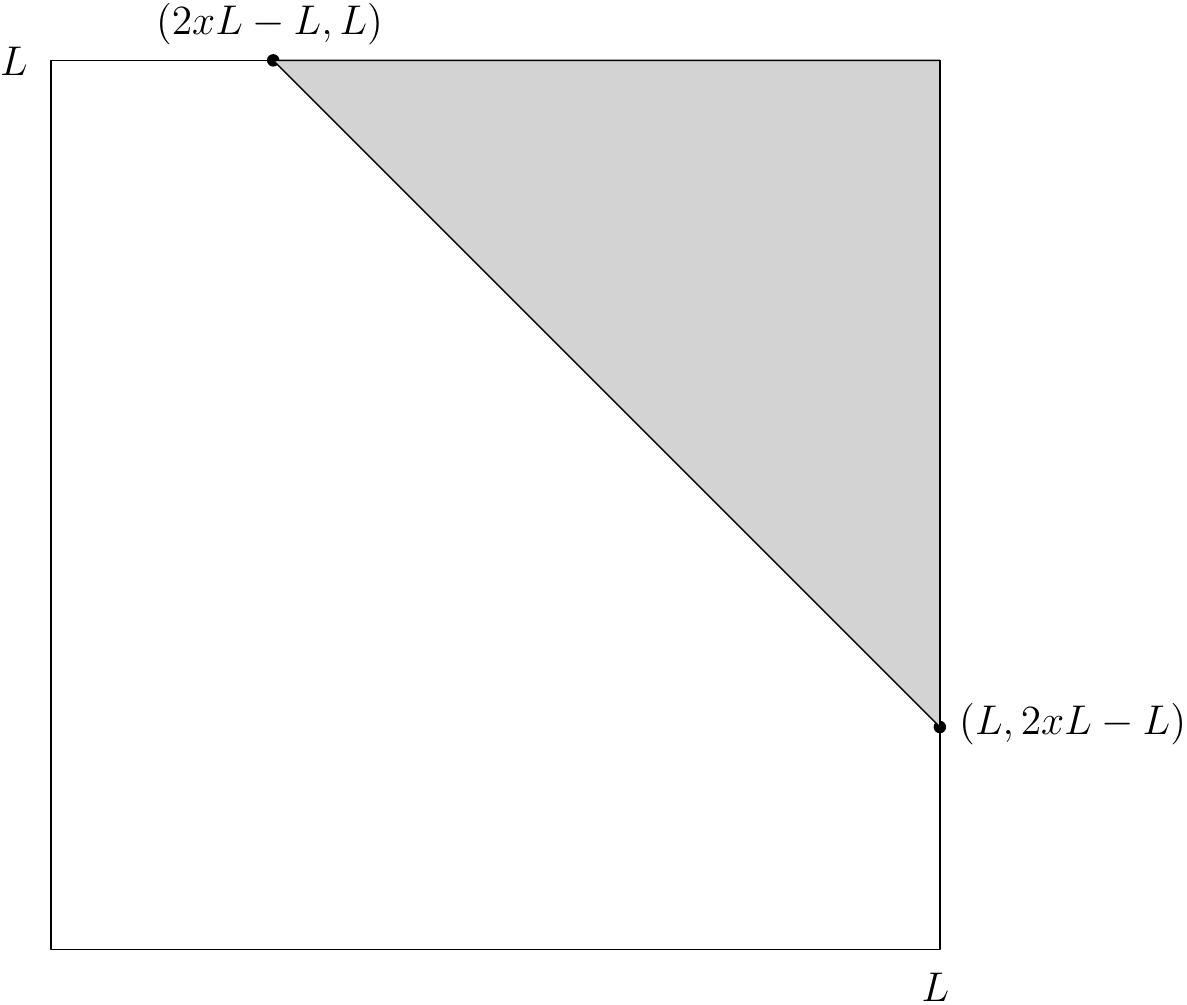}
  \caption{The square $[0, L] \times [0, L]$ is once again representing the possible pairs of lengths for the two intervals. The gray region now illustrates the cases with an overlap given $C = xL$ for $0.5 < x \leq 1$}
  \label{fig:P_ov_given_center_case_B}
\end{figure}
%\hfill \qed
\end{proof}
We can see that the risk of an overlap is close to $0$ if the centers are close to each other and that the risk is highest if the distance between the centers is $\frac{1}{3}L$. The risk decreases for higher distances between the centers than $\frac{1}{3}L$ and reaches (not surprisingly) $0$ for distances above $L$.

Even though our algorithm is oblivious to information on other items than the entering item, it has a very low risk of assigning the same storage location to two conflicting items compared to a random stacking strategy. The risk is of the order $1/k^2$ assuming that the instances are generated using the Scheinerman model. Our next lemma expresses the exact risk where $SC$ denotes the event that two intervals receive the same color ($=$ storage location) by our algorithm.

\begin{lemma}\label{lemma:P_SC_given_OV} For two intervals drawn using the Scheinerman model with $L$ satisfying (\ref{eq:L_assumption}) the following holds for $k \geq 3$:
$$\Pr(SC \mid OV) = \frac{12}{8 - 3L} \left(\frac{4}{3(k-1)^2}-\frac{1}{(k-1)^3}\right) \enspace .$$
\end{lemma}

\begin{proof}
It is well known and straightforward to show that the probability density function for the distance between two random numbers chosen uniformly at random in $[0, a]$ is as follows:
%$$F(x) = 1-\frac{(a-x)^2}{a^2}$$
\begin{equation}
\label{eq:pdf_distance}
g_a(x) = \frac{2}{a^2}\left(a-x\right) ,\enspace 0 \leq x \leq a \enspace .
\end{equation}

Let $SI$ denote the event that the centers of the intervals belong to the same $J$-interval. There are $\frac{1}{k} \cdot \frac{k-1}{L}$ intervals for each color implying
$$\Pr(SI \mid SC) = \frac{k}{k-1} L \enspace .$$
The maximum length of an interval is $L$ so $\Pr(OV \cap \overline{SI} \mid SC)=0$ where $\overline{SI}$ is the complementary event of the event $SI$. This means that
$$\Pr(OV \mid SC) = \Pr(OV \cap SI \mid SC) \enspace .$$
The event $SI$ implies the event $SC$ ($SI \cap SC = SI$):
$$\Pr(OV \cap SI \mid SC) = \Pr(OV \mid SI) \cdot  \Pr(SI \mid SC) \enspace .$$
We can now apply Lemma~\ref{lemma:P_ov_given_center} and (\ref{eq:pdf_distance}) and use the law of total probability:
%$$\Pr(OV \mid S) \leq (4/k-2/k^2) \cdot 1/(1/L) = (4/k-2/k^2)L $$
\begin{equation*}
\begin{split}
\Pr(OV \mid SC) & =  \Pr(OV \mid SI) \cdot  \Pr(SI \mid SC) \\ 
& = \int_0^{\frac{1}{k-1}} g_{\frac{1}{k-1}}(x)\Pr(OV \mid C=xL) dx \cdot \Pr(SI \mid SC) \\
& = \int_0^{\frac{1}{k-1}} 2(k-1)^2\left(\frac{1}{k-1}-x\right)(4x - 6x^2) dx \cdot \frac{k}{k-1}L \\
%& = kL\int_0^{\frac{1}{k-1}} 2\left(1-(k-1)x\right)(4x - 6x^2) dx \\
%& = kL\int_0^{\frac{1}{k-1}} 2\left(4x-6x^2-4(k-1)x^2+6(k-1)x^3\right) dx \\
& = kL \left[4x^2-4x^3-\frac{8}{3}(k-1)x^3+3(k-1)x^4\right]_0^{\frac{1}{k-1}} \\
%& = kL \left(\frac{4}{(k-1)^2}-\frac{4}{(k-1)^3}-\frac{8}{3(k-1)^2}+\frac{3}{(k-1)^3}\right) \\
& = kL \left(\frac{4}{3(k-1)^2}-\frac{1}{(k-1)^3}\right) \enspace .
\end{split}
\end{equation*}
We once again use the law of total probability and Lemma~\ref{lemma:P_ov_given_center} and compute the probability that two intervals overlap. This time we use $g_1(x)$ from (\ref{eq:pdf_distance}):
%$$\Pr(OV) = \frac{2}{3} L \mbox{ (see below)}$$
\begin{equation*}
\begin{split}
\Pr(OV) & = \int_{0}^{1} g_1(x)\Pr(OV \mid C=x) dx \\
& = \int_{0}^{1/2L} \left(4 \cdot \frac{x}{L} - 6\left(\frac{x}{L}\right)^2\right)2(1-x) dx + \int_{1/2L}^{L} \frac{1}{2}\left(2-2 \cdot \frac{x}{L}\right)^2 2(1-x) dx \\
& = \int_{0}^{1/2} L\left(4t - 6t^2\right)2(1-Lt) dt + \int_{1/2}^{1} L\frac{1}{2}\left(2-2t\right)^2 2(1-Lt) dt \\
%& = L\left(\int_{0}^{1/2} 2\left(4t - 6t^2\right) dt + \int_{1/2}^{1} \left(2-2t\right)^2  dt\right) - L^2\left(\int_{0}^{1/2} \left(4t - 6t^2\right)2t\ dt + \int_{1/2}^{1} \left(2-2t\right)^2 t\ dt\right)\\
%& = \frac{2}{3}L - L^2\left(\int_{0}^{1/2} \left(4t - 6t^2\right)2t\ dt + \int_{1/2}^{1} \left(2-2t\right)^2 t\ dt\right)\\
%& = \frac{2}{3}L - L^2\left(\left[\frac{8}{3}t^3-3t^4\right]_0^{1/2} + \int_{1/2}^{1} \left(4t^3-8t^2+4t\right) dt\right)\\
%& = \frac{2}{3}L - L^2\left(\left[\frac{8}{3}t^3-3t^4\right]_0^{1/2} + \left[t^4-\frac{8}{3}t^3+2t^2\right]_{1/2}^{1}\right)\\
%& = \frac{2}{3}L - L^2\left(\frac{1}{3} - \frac{3}{16} + 1 - \frac{8}{3} + 2 - \frac{1}{16} + \frac{1}{3} - \frac{1}{2}\right) \\
%& = \frac{2}{3}L - L^2\left(- \frac{3}{16} + 1 - \frac{1}{16} - \frac{1}{2}\right) \\
%& = \frac{2}{3}L - L^2\left(\frac{-3+16-1-8}{16}\right) \\
& = \frac{2}{3}L - \frac{1}{4}L^2 \enspace .
\end{split}
\end{equation*}
The assumption (\ref{eq:L_assumption}) on $L$ implies
$$\Pr(SC) = \frac{1}{k} \enspace .$$
%$$\Pr(S \mid OV) = \frac{\Pr(OV \mid S)\Pr(S)}{\Pr(OV)} \leq 2/k(4/k-2/k^2) = 4/k^2 \cdot (2-\frac{1}{k}) $$
To prove the lemma, we now apply Bayes Theorem:
\begin{equation*}
\begin{split}
\Pr(SC \mid OV) & = \frac{\Pr(OV \mid SC)\Pr(SC)}{\Pr(OV)} \\
& = \frac{L}{\frac{2}{3}L - \frac{1}{4}L^2} \left(\frac{4}{3(k-1)^2}-\frac{1}{(k-1)^3}\right)\\
& = \frac{12}{8 - 3L} \left(\frac{4}{3(k-1)^2}-\frac{1}{(k-1)^3}\right) \enspace .
\end{split}
\end{equation*}
%\hfill \qed
\end{proof}

Lemma~\ref{lemma:P_SC_given_OV} is the core lemma for our paper so the lemma was verified empirically by simulation in Python using the pseudo-random number generator from the module "random". We generated $n=200000$ random intervals using the Scheinerman model and registered whether the intervals were overlapping and/or received the same color by our algorithm. In this way an empirical estimate of the probability was computed. Four simulations were performed for $k \in \{5, 10, 20, 30\}$ with $L = \frac{k-1}{5k}$. The lemma was confirmed by the following relative differences\footnote{$\frac{\mbox{estimate}-\mbox{probability}}{\mbox{probability}}$} between the empirical estimates and the probability stated by the lemma for the four values of $k$, respectively (the probability was used as the reference value): $-0.17\%$, $-0.22\%$, $-0.01\%$, and $0.48\%$.

We are now ready to present the main theorem of our paper for $k \geq 3$.
\begin{theorem}\label{thm:scheinerman_model} Let a set of intervals $\mathcal{I}_n$ be drawn using the Scheinerman model with $L$ satisfying (\ref{eq:L_assumption}). The algorithm~(\ref{eq:algorithm}) computes a cut of the corresponding MAX $k$-CUT STACKING instance for $k \geq 3$ such that
$$\frac{E(|cut|)}{E(m)} = 1-\frac{12}{8 - 3L} \left(\frac{4}{3(k-1)^2}-\frac{1}{(k-1)^3}\right)$$
where $m$ denotes the number of edges in the corresponding interval overlap graph/circle graph.
\end{theorem} 

\begin{proof}
By using linearity of expectation we obtain the following:
\begin{equation*}
\begin{split}
\frac{E(|cut|)}{E(m)} & = \frac{\Pr(OV \cap \overline{SC}){n \choose 2}}{\Pr(OV){n \choose 2}} \\
& = \frac{\Pr(OV)-\Pr(OV \cap SC)}{\Pr(OV)} \\
& = 1 - \Pr(SC \mid OV)  \enspace .
\end{split}
\end{equation*}
The theorem follows from Lemma~\ref{lemma:P_SC_given_OV}.
%\hfill \qed
\end{proof}

\section{Analysis of the Extended Scheinerman Model}\label{sec:extended_Scheinerman}

Our results are now extended to a more generic model for generating the instances where we allow the lengths of the intervals to be drawn using any (fixed) bounded continuous probability density function with $L$ as an upper bound on the lengths. In this section we show that the risk of assigning the same color to two overlapping intervals is also of the order $1/k^2$ for the generic instance model.

\begin{definition}\label{def:Extended_Scheinerman} The Extended Scheinerman Model: The centers of the intervals in $\mathcal{I}_n$ are drawn independently using a uniform distribution on $[0, 1]$. The number $L$ is an upper bound on the lengths of the intervals and the lengths are drawn independently using a bounded continuous probability density function $f$, $f(\ell) \leq B$.
\end{definition}

\begin{theorem}\label{thm:extended_scheinerman_model}
Under the same assumptions as in Theorem~\ref{thm:scheinerman_model} the following holds for algorithm~(\ref{eq:algorithm}) for the extended Scheinerman model:
$$\frac{E(|cut|)}{E(m)} \geq 1-O(k^{-2}) \enspace .$$
\end{theorem}

\begin{proof}
Now assume that the intervals are drawn using the extended Scheinerman model.
By a slight modification of the proof of Lemma~\ref{lemma:P_ov_given_center}, we get the following for $x \leq 0.5$:
$$Pr(OV \mid C=xL) \leq L^2B^2(4x - 6x^2) \enspace .$$

We revisit the proof of Lemma~\ref{lemma:P_SC_given_OV} and establish an upper bound for $\Pr(OV~\mid~SC)$ using the factor $L^2B^2$:
$$\Pr(OV \mid SC) \leq L^2B^2 \cdot kL \left(\frac{4}{3(k-1)^2}-\frac{1}{(k-1)^3}\right) \enspace .$$

The probability of receiving the same color has not changed, $\Pr(SC) = \frac{1}{k}$, so Bayes Theorem can once again ensure that the theorem holds. We just have to make sure that the conditional probability is well defined -- in other words that $\Pr(OV) > 0$: There exists an $\ell_0$ such that $f(\ell_0) > 0$. Let $\epsilon > 0$ be a sufficiently small positive number. Let $\ell_1$ and $\ell_2$ denote the lengths of two intervals drawn using the extended Scheinerman model. By using (\ref{eq:overlap_condition}) and independence we can verify that $\Pr(OV) > 0$:
$$\Pr(OV) \geq \Pr(|l_1 - l_0| < \epsilon \wedge |l_2 - l_0| < \epsilon \wedge \epsilon < d < \ell_0-\epsilon) > 0 \enspace .$$ This concludes the proof.
%\hfill \qed
\end{proof}

The theorem implies a corollary that is targeted at an audience with an interest in circle graphs/interval overlap graphs. The corollary indicates that these graphs on average have $k$-cuts with a relatively high cardinality compared to the number of edges.
\begin{corollary} Let $\mathcal{I}_n$ be drawn using the extended Scheinerman model. For the circle graph/interval overlap graph represented by $\mathcal{I}_n$ we have the following for $k \geq 3$
$$\frac{E(|cut|)}{E(m)} \geq 1-O(k^{-2})$$
where "$cut$" denotes the maximum size $k$-cut.
\end{corollary}

\section*{Conclusion}

We have presented an extremely simple oblivious stacking algorithm with a surprisingly low risk -- of the order $1/k^2$ -- of assigning two conflicting items to the same storage location under mild assumptions on the distribution of the storage time intervals. The algorithm can easily be used in a distributed setting with disconnected stacking agents. The principle guiding the algorithm is to assign the same storage location to two items if the centers of their storage time intervals are close to each other implying a low risk of a conflict. This principle can probably be used in other stacking algorithms/heuristics.

Our algorithm can be used on its own but there is also a possibility that it can be used in a preprocessing step since the items are split into $k$ groups with only a few conflicts to handle in each group. As an example, each group could correspond to an area of a container terminal with $k << n$. Our algorithm could be used to assign an area to each container in the preprocessing step and a local stacking algorithm could pick a specific storage location for the container in the particular area.

\bibliographystyle{plain} 

\bibliography{MinShift}

\end{document}